\documentclass[11pt,a4paper]{article}

\usepackage{amssymb}
\usepackage{latexsym}
\usepackage{amsmath}
\usepackage{graphicx}
\usepackage{amsthm}
\usepackage{empheq}
\usepackage{bm}
\usepackage{booktabs}
\usepackage[dvipsnames]{xcolor}
\usepackage{pagecolor}
\usepackage{subcaption}
\usepackage{tikz,pgfplots,lipsum,lmodern}
\usepackage{enumitem}

\usepackage[most]{tcolorbox}

\usepackage[margin=2.8cm]{geometry}

\usepackage{hyperref}
\usepackage[capitalize]{cleveref}

\numberwithin{equation}{section}

\theoremstyle{definition}
\newtheorem{theorem}{Theorem}[section]
\newtheorem{corollary}[theorem]{Corollary}
\newtheorem{proposition}[theorem]{Proposition}
\newtheorem{definition}[theorem]{Definition}
\newtheorem{example}[theorem]{Example}

\newtheorem{notation}[theorem]{Notation}
\newtheorem{remark}[theorem]{Remark}
\newtheorem{lemma}[theorem]{Lemma}

\usepackage{calrsfs}

\newcommand{\numberset}{\mathbb}

\newcommand{\Z}{\numberset{Z}}

\newcommand{\R}{\numberset{R}}
\newcommand{\C}{\mathcal{C}}

\newcommand{\F}{\numberset{F}}

\newcommand{\fq}{\F_q}

\newcommand{\mV}{\mathcal{V}}

\newcommand{\mC}{\mathcal{C}}

\newcommand{\mD}{\mathcal{D}}

\newcommand{\wH}{\omega^{\textnormal{H}}}

\newcommand\red[1]{{{\textcolor{red}{#1}}}}

\newtheorem{claim}{Claim}

\newcommand*{\myproofname}{Proof of the claim}
\newenvironment{clproof}[1][\myproofname]{\begin{proof}[#1]}{\end{proof}}

\usepackage{enumitem}
\setitemize{itemsep=0pt}
\setenumerate{itemsep=0pt}

\title{\textbf{LRCs: Duality, LP Bounds, and Field Size}}
\usepackage{authblk}

\author{Anina Gruica,
Benjamin Jany,
Alberto Ravagnani,
}
\affil{Department of Mathematics and Computer Science \\ Eindhoven University of 
	Technology, the Netherlands}

\date{}

\usepackage{setspace}
\allowdisplaybreaks

\begin{document}

\maketitle

\begin{abstract}
We develop a duality theory of locally recoverable codes (LRCs) and apply it to establish a series of new bounds on their parameters.
We introduce and study a refined notion of weight distribution that captures the code's locality.
Using a duality result analogous to a MacWilliams identity, we then derive an LP-type bound that improves on the best known bounds in several instances. Using a dual distance bound and the theory of generalized weights, we obtain non-existence results for optimal LRCs over small fields.
In particular, we show that an optimal LRC must have both minimum distance and block length relatively small compared to the field size.
\end{abstract}


\medskip
\section{Introduction}

In the last decade, \textit{locally recoverable codes} have been a central topic in communication and distributed storage~\cite{cadambe2015bounds, freij2018matroid,Gopalan, hao2020bounds, tamo2014family}. A code has good locality properties when each entry of each codeword can be recovered from a small set of other entries of the same codeword. This property 
is captured by a parameter of the underlying code called \textit{locality}.
Small locality 
allows, for instance, a fast recovery process when the code is used for distributed storage. 
Ideally, a code has both small locality and large minimum distance,
as the latter offers protection from 
errors and erasures.

Most of the research on locally recoverable codes focuses on bounds~\cite{agarwal2018combinatorial,cadambe2015bounds,Gopalan,guruswami19,hao2020bounds}, constructions and decoding~\cite{barg2017locally2,barg2017locally, Micheli,tamo2014family}. This paper pertains to the former line of research.

As one expects, a code cannot have all the desirable properties at the same time. In particular, it cannot have arbitrarily small locality while having large minimum distance for the given dimension and block length. The 
trade-offs among all these parameters are captured by various bounds, the best known of which is probably the \textit{Generalized Singleton Bound}~\cite{Gopalan}. Codes attaining this bound with equality are known to exist for some parameter sets over large fields. Over small fields,
it is still a wide open questions which parameters locally recoverable codes can have.

In this paper, we investigate two aspects of the theory of locally recoverable codes that have overall been neglected so far, 
namely their \textit{duality theory} and 
\textit{field size}.

We first demonstrate that duality results, besides being mathematically 
interesting, represent a very powerful tool for investigating the parameters of locally recoverable codes and establishing new bounds.
In the second part of the paper 
we derive various results that link the size of the underlying field to the other parameters of a locally recoverable code. In particular, we prove that codes attaining the 
Generalized Singleton Bound of~\cite{Gopalan} cannot exist for some parameter sets.

The rest of the introduction briefly illustrates the contributions made by this paper, pointing the reader to the relevant sections.

After recalling the basics of locally recoverable codes in Section~\ref{sec:lrc}, 
we introduce a refined notion of weight distribution of a code, which is able to capture its locality as well as the weights of the codewords; see Section~\ref{sec:finerweight}. In the same section, we then establish a duality result for the refined weight distribution, which is similar to, and yet different from,
a MacWilliams identity. The identity combined with linear programming produces a bound on the parameters of a locally recoverable code that improves the best bounds currently available for several parameter sets; see Section~\ref{sec:Applications}. We illustrate this with comparison tables.

The last part of the paper, namely Section~\ref{sec:fieldsize}, is devoted to the role played by the field size in the theory of locally recoverable codes. We propose various arguments based on the notions of dual distance and generalized weights. This leads us to establishing new bounds for the parameters of a locally recoverable code, which in turn gives us non-existence results for codes meeting the Generalized Singleton Bound over small fields.

\medskip

\paragraph*{Acknowledgements.}
We are grateful to 
Markus Grassl for discussing with us the connection between the results of this paper and the split weight enumerator of a code. A. G. is supported by the Dutch Research Council through grant OCENW.KLEIN.539. 
B. J. is supported by the Dutch Research Council through grant VI.Vidi.203.045.
A. R. is supported by the Dutch Research Council through grants VI.Vidi.203.045, 
OCENW.KLEIN.539, 
and by the Royal Academy of Arts and Sciences of the Netherlands.

\bigskip

\section{Locally Recoverable Codes} \label{sec:lrc}

We start by recalling 
concepts of classical coding theory and locally recoverable codes, and by establishing the notation. Throughout the paper, $n \ge 2$ is an integer, $q$ denotes a prime power, and $\F_q$ is the finite field of $q$ elements. We will denote the set $\{1, \dots, n\}$ by $[n]$.

\begin{definition}
A (\textbf{linear}) \textbf{code} is an $\F_q$-linear subspace $\mC \leq \F_q^n$  endowed with the Hamming metric. 
The \textbf{dual} of $\mC$ is $\mC^\perp :=\{x \in \F_q^n : x\cdot y^\top = 0 \textnormal{ for all } y \in \mC\}.$ Moreover, we say that $\mC \le \F_q^n$ is \textbf{non-degenerate} if there is no $i \in [n]$ for which each $x \in \mC$ has $x_i=0$.
The codes $\{0\}$ and $\F_q^n$ are called
\textbf{trivial}.
\end{definition}

In this document, by ``code'' we always mean ``non-trivial code'', unless otherwise specified.

\begin{definition}
The \textbf{minimum distance} of
a code $\mC$ is defined as 
\begin{align*}
    d(\mC) := \min\{\wH(x) : x \in \mC \setminus \{0\}\}, 
\end{align*}
where $\wH(x) :=|\{i \, : \, x_i \neq 0\}|$ is the (\textbf{Hamming}) \textbf{weight} of $x$.
\end{definition}

A large minimum distance guarantees the error detection and correction  capabilities of a code, so naturally one wants this parameter to be as large as possible.
However, as it is well known, there is a trade-off between the minimum distance and the dimension of a code with given length. This trade-off is expressed by the famous \emph{Singleton Bound}, which says that for a code $\mC \le \F_q^n$ we always have $d(\mC) \leq n - \dim_{\F_q}(\mC) +1$; see \cite{singleton1964maximum}. 
Codes that meet the bound with equality are called \textbf{MDS}. It remains a wide open problem in coding theory to determine when codes achieving this bound (or other bounds) exist. 

Apart from the minimum distance and the dimension, in this paper we consider an additional parameter for linear codes, called 
\emph{locality}.

\begin{definition}\label{def:locality}
    A code $\mC \leq \F_q^n$ has \textbf{locality} $r$ if for every $i \in \{1, \ldots, n\}$ there exists a set $S_i \subseteq [n]$, called a \textbf{recovery set} for (the coordinate) $i$, with the following properties:
    \begin{itemize}
        \item[1)] $i \notin S_i$,
        \item[2)] $|S_i| \leq r$,
        \item[3)] if $x, y \in \mC$ and $x_j = y_j$ for all $j \in S_i$, then $x_i = y_i$. 
    \end{itemize}
We call \textbf{LRC} (\textbf{locally recoverable code}) a code for which
the locality parameter is considered.
\end{definition}

\begin{example} \label{ex:simplex}
Let $\mC \le \F_2^7$ be the binary simplex code of dimension 3, which consists of codewords of the form
\begin{align*}
    (u_1,u_2,u_3)\begin{pmatrix} 
    1 & 0 & 0 & 1 & 0 & 1 & 1 \\
    0 & 1 & 0 & 1 & 1 & 0 & 1 \\
    0 & 0 & 1 & 0 & 1 & 1 & 1 
\end{pmatrix} =
(u_1,u_2,u_3,u_1+u_2,u_2+u_3,u_1+u_3,u_1+u_2+u_3),
\end{align*}
where $u_1,u_2,u_3 \in \F_2$.
Without loss of generality,
we focus on recovering coordinate $i=1$. It is not hard to check that recovery sets for 1 are
\begin{align*}
   \{2,4\},\, \{3,6\},\, \{5,7\},\, \{2,3,7\},\, \{4,6,7\},\, \{3,4,5\},\, \{2,5,6\}
\end{align*}
as well as any subset of $\{2,\dots,7\}$ containing one of the above.
By symmetry of the coordinates, this code has (minimum) locality 2.
\end{example}

As the name suggests, a recovery set $S_i$ for $i$ allows to recover the coordinate $x_i$ of any codeword $x \in \mC$ using only the coordinates of $x$ indexed by $S_i$. 
This is done via a \textbf{recovery function} $f_i : \pi_{S_i}(\mC) \rightarrow \F_q$ that satisfies $f_i(\pi_{S_i}(x)) = x_i$ for all $x \in \mC$, where $\pi_{S_i}$ denotes the projection map 
onto the coordinates indexed by the elements of $S_i$. Interestingly, for linear codes the recovery functions must be linear; see~\cite[Lemma 10]{agarwal2018combinatorial}.

\begin{proposition}\label{thm:linearfct}
Let $\mC \le \F_q^n$ be a code and let $S_i$ be a recovery set for the coordinate~$i$ with recovery function $f_i$. Then $f_i$ is an $\F_q$-linear map. 
\end{proposition}

Similarly to classical codes, it is important to understand the trade-offs between the locality parameter and other parameters of the code, such as its minimum distance and dimension.
In~\cite{Gopalan}, the authors showed how locality impacts the dimension of an LRC by establishing a generalization of the Singleton Bound. 

\begin{theorem}[Generalized Singleton Bound] \label{thm:single}
Let $\mC \le \F_q^n$ be a code with locality $r$, dimension $k$, and minimum distance~$d$. Then 
\begin{equation}\label{eqt:single}
    k + \left \lceil \frac{k}{r} \right \rceil \leq n - d+2.
\end{equation}
\end{theorem}

Note that the bound of Theorem~\ref{thm:single} coincides with the classical Singleton Bound if $k=r$. Codes whose parameters meet the bound (\ref{eqt:single}) with equality are called \textbf{optimal LRC}. Note that the simplex code in Example~\ref{ex:simplex} is optimal. 

In~\cite{tamo2014family}, a construction of optimal LRC was given for $q \geq n$, $r \mid k$, and $r+1 \mid n$. 
However,  if these divisibility constraints are not satisfied or if $q < n$, optimal LRC do not always exist or have not been found yet. In order to establish or exclude the existence of optimal LRCs for small~$q$, it 
is natural 
to establish bounds
linking the locality 
to the size of the underlying field. The following is a shortening bound established in~\cite{cadambe2015bounds}, that improves the Singleton-type bound in Theorem~\ref{thm:single}. In the minimum we include the (trivial) case $t=0$, even though the original statement does not. Note that for some parameters the minimum is indeed attained by $t=0$.

\begin{theorem}[Shortening Bound]\label{thm:kopt}
Let $\mC \le \F_q^n$ be a code with locality $r$, dimension $k$ and minimum distance~$d$. We have
\begin{equation}\label{eqt:kopt}
    k \le \min_{t \in \Z_{\ge 0}}\left\{rt+k_{\textnormal{opt}}^{(q)}(n-t(r+1),d)\right\},
\end{equation}
where $\smash{k_{\textnormal{opt}}^{(q)}(n,d)}$ is the largest possible dimension of a code of length $n$ and minimum distance~$d$ over $\F_q$.
\end{theorem}

\begin{remark} \label{rem:compdraw}
\begin{itemize}
    \item[(i)] It was observed in \cite{hao2020bounds} that \eqref{eqt:kopt} yields a series of bounds on codes with locality, by applying known bounds on  \smash{$k_{\textnormal{opt}}^{(q)}(n-t(r+1),d)$}. For example, 
    by letting \smash{$t = \left \lceil k/r \right \rceil$} and using the classical Singleton Bound on \smash{$k_{\textnormal{opt}}^{(q)}(n-t(r+1),d)$}, we recover Theorem~\ref{thm:single}.
    \item[(ii)] Even though
    the bound of Theorem~\ref{thm:kopt} is a refinement of Theorem~\ref{thm:single},
    there is a ``computational'' drawback in evaluating it.
    Indeed, determining the value of~$k_{\textnormal{opt}}^{(q)}(n,d)$ for given $d, n, q$ is a wide open problem in classical coding theory, which means that no closed formula is known for the
    RHS of \eqref{eqt:kopt}. Establishing 
    bounds whose evaluation is
    computationally feasible
    has therefore become a crucial research problem in the study of LRCs. This is one of the problems this paper addresses.
\end{itemize}
\end{remark}

The concept of locally recoverable codes was further generalized in~\cite{prakash2012optimal}, where
the authors introduce the notion of $(r, \delta)$-LRC.  The additional parameter $\delta$  provides extra information regarding the number of recovery sets for coordinate~$i$ of size less than or equal to $r$. In the  applications, $(r, \delta)$-LRCs facilitate the local recovery of a failed node in the event that other nodes of the network fail as well.

\begin{definition} \label{def:rdeltaloc}
A non-degenerate code $\mC \leq \F_q^n$ has \textbf{locality} $(r, \delta)$ (or is \textbf{$(r,\delta)$-LRC}) if for all $i \in [n]$ there exists a set $S_i \subseteq [n]$ such that:
\begin{itemize}
    \item[1)] $i \notin S_i$,
    \item[2)] $|S_i| \leq r+\delta $,
    \item[3)] $d(\pi_{S_i \cup \{i\}}(\mC)) \geq \delta$.
\end{itemize}
We then call $S_i$ an \textbf{$(r,\delta)$-recovery set} for $i$.
\end{definition}

Note that for an $(r,\delta)$-LRC $\mC \le \F_q^n$, an $(r,\delta)$-recovery set $S_i$ for $i \in [n]$, and $T \subseteq S_i \cup \{i\}$ with $|T|=\delta-1$, the coordinates of any $x \in \mC$ indexed by $T$ can be recovered from the coordinates of $x$ indexed by $\left(S_i\cup \{i\} \right) \setminus T$. This follows from Definition~\ref{def:rdeltaloc}, part 3). Moreover, it can easily be seen that the notions of $(r,\delta)$-LRC and LRC with locality~$r$ coincide when $\delta=2$.

A refinement of 
Theorem~\ref{thm:single} taking the parameter $\delta$ into account was established in \cite{prakash2012optimal} and it reads as follows:
\begin{equation}\label{eqt:Singlrdelta}
    d \leq n - k +1 - \left( \left \lceil \frac{k}{r} \right \rceil -1\right)(\delta -1).
\end{equation}

Furthermore, the following is a generalization of \cref{thm:kopt} which was established in~\cite{grezet2019alphabet,rawat2015cooperative}:
\begin{equation}\label{eq:rdeltakopt}
    k \leq \min_{t \in \Z_{\geq 0}} \{tr + k_{\textnormal{opt}}^{(q)}(n-t(r+\delta-1),d)\}. 
\end{equation}

Similarly to \cref{thm:kopt}, the bound of \eqref{eq:rdeltakopt} has a computational drawback, as already explained in Remark~\ref{rem:compdraw}.

Note that both \eqref{eqt:single} and \eqref{eqt:kopt} can be recovered by the two previous bounds when $\delta = 2$.

In this paper, we focus mostly on classical LRCs (i.e. $\delta=2)$. However, in some instances we consider the broader case of $(r, \delta)$-LRCs for any $\delta$, when the techniques we develop are applicable.

\section{A Finer Weight Distribution and Duality}\label{sec:finerweight}

In this section, we introduce and work with a refinement of the classical \textit{weight distribution} of a code. Recall that the \textbf{weight distribution} of $\mC \le \F_q^n$ is the tuple $(W_0(\mC), \dots, W_n(\mC))$, where
\begin{align*}
    W_i(\mC) := |\{x \in \mC : \wH(x) = i\}|.
\end{align*}

Before we introduce the aforementioned refinement of the weight distribution, we show how the locality of a code $\mC \leq \F_q^n$ can be captured by looking at the support of codewords in the dual code $\mC^{\perp}$. Recall that for a vector $x \in \F_q^n$, its \textbf{support} is defined to be the set $\sigma(x) := \{ i \mid x_i \neq 0\}$. We now have the following equivalent characterization of the locality parameter, which will be used extensively in this paper.

\begin{lemma}\label{local/sup}
Let $r \ge 1$ be an integer. A linear code $\mC \le \F_q^n$ has locality $r$ if and only if for any $i \in [n]$ there exists $x \in \mC^{\perp}$ with $i \in \sigma(x)$ and $\wH(x) \le r+1$.
\end{lemma}

Lemma~\ref{local/sup} was first established in  \cite[Lemma 5]{guruswami19} and heavily relies on the fact that recovery functions for linear codes are always linear; see Proposition~\ref{thm:linearfct}. The characterization of Lemma~\ref{local/sup} inspires the following definition, where instead of caring only about how many codewords of a certain weight there are (which is captured in the classical weight distribution), we want to know the number of codewords of a certain weight having a certain subset of $[n]$ in their support.

\begin{definition} \label{def:wis}
For a code $\mC \le \F_q^n$, a set $S \subseteq [n]$, and $0 \leq i \leq n$, we let 
\begin{align*}
    W_i^S(\mC) := |\{x \in \mC : \wH(x) = i, \, S \subseteq \sigma(x)\}|.
\end{align*}
\end{definition}

Note that if we set $S=\emptyset$, then the refined weight distribution introduced in Definition~\ref{def:wis} recovers the classical weight distribution. The main result of this section  is a MacWilliams-type identity for the refined weight distribution. In more detail, we will show that the refined weight distribution of the dual code fully determines the refined weight distribution of the original code. 

We start with the following result, which relates values of the refined weight distributions to one another.  

\begin{lemma}\label{lem:countS}
Let $\mC \leq \F_q^n$ and $A \subseteq [n]$. Then for all $1 \leq i \leq n$ and $|A| \leq t \leq i$ we have  
$$W_i^A(\mC) = \frac{1}{{\binom{i-|A|}{t- |A|}}}\sum_{\substack{A \subseteq S \subseteq [n] \\ |S| = t}^{\,}} W_i^S(\mC).$$
\end{lemma}

\begin{proof}
Let $\mV = \{ (x, S) \mid x \in \mC, \, \wH(x) = i, \, A \subseteq S \subseteq \sigma(x), \, |S| = t\}$. We count the elements of $\mV$ in two ways. On one hand, we have  
\begin{align} \label{eq:lem331}
    |\mV|= \sum_{\substack{A \subseteq S \subseteq [n] \\ |S| = t}^{\,}} W_i^S(\mC).
\end{align}
On the other hand, 
\begin{align} \label{eq:lem332}
    |\mV| = \sum_{\substack{x \in \mC \\ \wH(x) = i \\ A \subseteq \sigma(x)}}| \{S \, : \, A \subseteq S \subseteq \sigma(x) \textup{ and } |S| = t\}|
        ={i -|A| \choose t - |A|} W_i^A(\mC). 
\end{align}
Combining~\eqref{eq:lem331} and~\eqref{eq:lem332} concludes the proof. 
\end{proof}

In order to obtain the MacWilliams-type identity for the refined weight distribution from Definition~\ref{def:wis}, we need to introduce auxiliary definitions and results.

\begin{definition} \label{def:cst}
For a code $\mC \le \F_q^n$ and subsets $S \subseteq T \subseteq [n]$, let
\begin{align*}
    \mC(S,T):=\{x \in \mC : S \subseteq \sigma(x), \, x \in \mC(T)\},
\end{align*}
where $\mC(T)$ denotes the \textbf{shortening}
of $\mC$ by the set $T$, i.e.,
$\mC(T)=\{x \in \mC : \sigma(x) \subseteq T\}$.
\end{definition}

We can now express the cardinality of a subcode introduced in Definition~\ref{def:cst} in terms of the cardinality of subcodes of the dual code.

\begin{proposition} \label{prop:cst} 
Let $\mC \le \F_q^n$. For all $S \subseteq T \subseteq [n]$ we have
\begin{align*}
    |\mC(S,T)| = |\mC| \sum_{A \subseteq S} (-1)^{|A|}\displaystyle\frac{|\mC^\perp(T^c \cup A)|}{q^{n-|T|+|A|}}.
\end{align*}
\end{proposition}
\begin{proof}
We have
\begin{align} \label{eq:lemcst1}
    |\mC(S,T)| &= |\mC(T)|- |\{x \in \mC : \sigma(x) \subseteq T, \, \sigma(x) \subseteq [n]\setminus A \textnormal{ for some $\emptyset \ne A \subseteq S$}\}| \nonumber \\
    &= |\mC(T)|- |\{x \in \mC : \sigma(x) \subseteq T \cap ([n]\setminus A) \textnormal{ for some $\emptyset \ne A \subseteq S$}\}|.
\end{align}
Since $A \subseteq S \subseteq T$, we can rewrite~\eqref{eq:lemcst1} as
\begin{align*}
    |\mC(T)|- |\{x \in \mC : \sigma(x) \subseteq T\setminus A \textnormal{ for some $\emptyset \ne A \subseteq S$}\}| &= |\mC(T)|- \left| \bigcup_{\emptyset \ne A \subseteq S} \mC(T \setminus A)\right| \\
    &= \sum_{A \subseteq S} (-1)^{|A|}\left|\mC(T \setminus A)\right|,
\end{align*}
where the latter equality follows
from the Inclusion-Exclusion principle. Finally, we have
\begin{align*}
    |\mC(T \setminus A)| = \frac{|\mC|}{q^{n-|T|+|A|}}|\mC^\perp(T^c \cup A)|,
\end{align*}
from which the statement of the lemma follows.
\end{proof}

With the aid of Proposition~\ref{prop:cst} we can prove the following result, which can be seen as an analogue of a \textit{MacWilliams binomial moment} identity. 

\begin{proposition}\label{cor:sumw_i^S}
Let $\mC \le \F_q^n$, $S \subseteq [n]$ and $|S| \le t \le n$. We have
\begin{multline*}
    \sum_{i=0}^n \binom{n-i}{t-i} W_i^S(\mC) = \\ q^{k-n+t-|S|}(q-1)^{|S|}\sum_{i=0}^n  \sum_{D \subseteq S}\sum_{B \subseteq D} (-1)^{|D|-|B|} (1 - q)^{-|B|} \binom{n-|S|-i+|B|}{t-|S|}W_i^D(\mC^{\perp}).
\end{multline*}
\end{proposition}
\begin{proof}
We prove the statement by first fixing a subset $S \subseteq [n]$ and summing both sides of the expression in Proposition~\ref{prop:cst}  over all $T \subseteq [n]$ with $S \subseteq T$ and $|T|=t$. The LHS of the equality in Proposition~\ref{prop:cst} gives
\begin{align*}
  \sum_{\substack{S \subseteq T \subseteq [n] \\ |T|=t}} |\mC(S,T)|  &= \sum_{\substack{x \in \mC, \\ \sigma(x) \supseteq S}} |\{T \subseteq [n] : |T|=t,  \, \sigma(x) \subseteq T, \, S \subseteq T\}| \\
  &= \sum_{i=0}^n\sum_{\substack{x \in \mC \\ \sigma(x) \supseteq S \\ \wH(x) = i}} |\{T \subseteq [n] : |T|=t, \, S \subseteq \sigma(x) \subseteq T\}| \\
  &= \sum_{i=0}^n\sum_{\substack{x \in \mC \\ \sigma(x) \supseteq S \\ \wH(x) = i}} \binom{n-i}{t-i} = \sum_{i=0}^n \binom{n -i}{t-i} W_i^S(\mC),
\end{align*}
which is the LHS of the corollary.
For the RHS we compute 
\begin{align} \label{eq:corwis2}
   \sum_{\substack{S \subseteq T \subseteq [n]\\ |T|=t}}|\mC| \sum_{A \subseteq S} (-1)^{|A|}\displaystyle\frac{|\mC^\perp(T^c \cup A)|}{q^{n-t+|A|}} = |\mC| \sum_{A \subseteq S} \displaystyle\frac{(-1)^{|A|}}{q^{n-t+|A|}} \sum_{\substack{S \subseteq T \subseteq [n]\\ |T|=t}}|\mC^\perp(T^c \cup A)|.
\end{align}
We now 
write the last sum in~\eqref{eq:corwis2} differently. For a fixed~$A \subseteq S$ we have
\begin{align*} 
    \sum_{\substack{S \subseteq T \subseteq [n] \\ |T|=t}}|\mC^\perp(T^c\cup A)| &= \sum_{i=0}^n \sum_{\substack{x \in \mC^\perp \\ \wH(x) = i}} |\{T \subseteq [n] : |T|  =t, \sigma(x) \subseteq T^c\cup A, S \subseteq T\}| \\ 
    &= \sum_{i=0}^n \sum_{\substack{x \in \mC^\perp \\ \wH(x) = i}} |\{T \subseteq [n] : |T|  =n-t, \sigma(x) \subseteq T\cup A, T \subseteq S^c\}|\\
    &= \sum_{i=0}^n \sum_{\substack{x \in \mC^\perp \\ \wH(x) = i \\ \sigma(x) \subseteq S^c \cup A}} |\{T \subseteq S^c : |T|  =n-t, \sigma(x) \subseteq T\cup A\}|\\
    &= \sum_{i=0}^n \sum_{j=0}^{|A|} \sum_{\substack{x \in \mC^\perp, \\\wH(x) = i \\ |\sigma(x) \cap A|=j \\ |\sigma(x) \cap S^c|=i-j }} \binom{n-|S| -i+j}{n-t-i+j} \\
    &=\sum_{i=0}^n \sum_{j=0}^{|A|} \sum_{\substack{B \subseteq A \\ |B| = j}}\, \sum_{\substack{D \subseteq S^c, \\ |D| = i-j}}\, \sum_{\substack{x \in \mC^\perp\\ \sigma(x)=D \cup B }} \binom{n-|S| -i+j}{n-t-i+j}\\
    &= \sum_{i=0}^n  \sum_{B \subseteq A}\sum_{\substack{D \subseteq S^c, \\ |D| = i-|B|}} \binom{n-|S| -|D|}{n-t-|D|} W_i^{D \cup B}(\mC^\perp),
\end{align*}
where the last step follows from the fact that $|\{x \in \mC^\perp: \sigma(x) = B \cup D\}| = W_i^{D \cup B}(\mC^\perp)$ for $B \subseteq A$ with $|B|=j$ and $D \subseteq S^c$ with $|D|=i-j$. After substituting this expression into~\eqref{eq:corwis2}, we simplify further by reordering the summands and applying the Binomial Theorem as follows:

\begin{align*}
 &|\mC| \sum_{A \subseteq S} \displaystyle\frac{(-1)^{|A|}}{q^{n-t+|A|}} \sum_{i=0}^n  \sum_{B \subseteq A}\sum_{\substack{D \subseteq S^c \\ |D| = i-|B|}} \binom{n-|S| -|D|}{n-t-|D|} W_i^{D \cup B}(\mC^\perp) \\
  &= q^{k-n+t}\sum_{i=0}^n \sum_{B \subseteq S}  \sum_{\substack{D \subseteq S^c\\ |D| = i-|B|}} \binom{n-|S| -|D|}{n-t-|D|}  W_i^{D \cup B}(\mC^\perp) \sum_{B \subseteq A \subseteq S} \displaystyle\frac{(-1)^{|A|}}{q^{|A|}}\\
&= q^{k-n+t}\sum_{i=0}^n \sum_{B \subseteq S}  \sum_{\substack{D \subseteq S^c\\ |D| = i-|B|}} \binom{n-|S| -|D|}{n-t-|D|}  W_i^{D \cup B}(\mC^\perp)\sum_{j=|B|}^{|S|} (-1)^j \binom{|S| - |B|}{j- |B|} q^{-j}\\
&=q^{k-n+t}\sum_{i=0}^n \sum_{B \subseteq S}  \sum_{\substack{D \subseteq S^c\\ |D| = i-|B|}} \binom{n-|S| -|D|}{n-t-|D|}  W_i^{D \cup B}(\mC^\perp) (-1)^{|S|} q^{-|B|}  \\*
& \hspace{8cm}\cdot\sum_{j=0}^{|S|-|B|} (-1)^{|S|- |B| - j} \binom{|S| - |B|}{j} (q^{-1})^j\\
&=q^{k-n+t}\sum_{i=0}^n \sum_{B \subseteq S}  \sum_{\substack{D \subseteq S^c\\ |D| = i-|B|}} \binom{n-|S| -|D|}{n-t-|D|}  W_i^{D \cup B}(\mC^\perp) (-1)^{|S|} q^{-|B|}  (q^{-1}-1)^{|S|-|B|}\\
&= q^{k-n+t-|S|}(q-1)^{|S|}\sum_{i=0}^n  \sum_{B \subseteq S} (1 - q)^{-|B|} \sum_{\substack{D \subseteq S^c\\ |D| = i-|B|}} \binom{n-|S| -i +|B|}{t-|S|}  W_i^{D \cup B}(\mC^{\perp}).
\end{align*}
To conclude the proof we will need the following claim.
\begin{claim}\label{cl:simplificationequation} 
Let $\mC \leq \F_q^n$, $S \subseteq [n]$ with $|S| = t$, $B \subseteq S$, and $1 \leq j \leq n$. Then 
$$\sum_{\substack{D \subseteq S^c \\ |D| = j- |B|}} W_j^{D \cup B}(\mC^{\perp}) = \sum_{B \subseteq D \subseteq S}(-1)^{|D| - |B|} W_j^D(\mC^{\perp}).$$
\end{claim}
\begin{clproof}
For all $D \subseteq S^c$ with $|D| = j- |B|$ we have $$W_j^{D \cup B}(\mC^{\perp}) =  \left| \{v \in \mC^{\perp} \, : \, \sigma(v) = D \cup B\} \right|.$$ Hence 
\begin{align*}
\sum_{\substack{D \subseteq S^c \\ |D| = j-|B|}}W_j^{D \cup B}(\mC^{\perp}) &= \left| \{v \in \mC^{\perp} \, : \, D \subseteq S^c, \,  |D| = j-|B|, \, \sigma(v) = D \cup B \} \right|\\ 
&= \left | \{v \in \mC^{\perp} \, :  \, B \subseteq D \subseteq [n]\setminus (S \setminus B), \, |D| = j, \, \sigma(v) = D \} \right| \\
&= \left | \{ v \in \mC^{\perp} \, : \,   B \subseteq D \subseteq [n], \, |D|=j, \, \sigma(v) = D\} \right | - \\
&\left| \bigcup_{\emptyset \subsetneq A \subseteq (S\setminus B)} \{v \in \mC^{\perp} \, : \,  B\cup A \subseteq D \subseteq [n], \, |D|=j, \, \sigma(v) = D\} \right|.
\end{align*}
Using the Inclusion-Exclusion principle followed by \cref{lem:countS}, we get
\begin{align*}\sum_{\substack{D \subseteq S^c \\ |D| = j-|B|}}W_j^{D \cup B}(\mC^{\perp}) 
&= \sum_{\emptyset \subseteq A \subseteq S\setminus B}(-1)^{|A|} \sum_{\substack{A \cup B \subseteq D \subseteq [n]\\ |D| = j}} W_j^D(\mC^{\perp})\\
&= \sum_{\emptyset \subseteq A \subseteq S \setminus B}(-1)^{|A|} W_j^{A \cup B}(\mC^{\perp})\\
&= \sum_{B \subseteq D \subseteq S}(-1)^{|D|-|B|}W_j^D(\mC^{\perp}),
\end{align*}
proving the desired claim.
\end{clproof}
Finally, from Claim~\ref{cl:simplificationequation} it follows that
\begin{multline*}
q^{k-n+t-|S|}(q-1)^{|S|}\sum_{i=0}^n  \sum_{B \subseteq S} (1 - q)^{-|B|} \sum_{\substack{D \subseteq S^c\\ |D| = i-|B|}} \binom{n-|S| -i +|B|}{t-|S|}  W_i^{D \cup B}(\mC^{\perp})
=\\
q^{k-n+t-|S|}(q-1)^{|S|}\sum_{i=0}^n  \sum_{D \subseteq S}\sum_{B \subseteq D} (-1)^{|D|-|B|} (1 - q)^{-|B|} \binom{n-|S|-i+|B|}{t-|S|}W_i^D(\mC^{\perp}),
\end{multline*}
concluding the proof.
\end{proof}

In order to ``isolate'' the refined weight distribution of a code $\mC \le \F_q^n$ in the sum of \cref{cor:sumw_i^S}, we will use the following lemma.

\begin{lemma}\label{lem:bomb}
Let $$\alpha_t:=\sum_{i=0}^n \binom{n-i}{t-i} \beta_i \quad \textnormal{for $0 \le t \le n.$}$$ Then $$\beta_i = \displaystyle\sum_{t=0}^n (-1)^{(i-t)} \binom{n-t}{i-t} \alpha_t \quad \textnormal{for all $0 \le i \le n.$}$$
\end{lemma}
\begin{proof}
Let $0 \le i \le n$ be fixed. We have
\begin{align*}
    \sum_{t=0}^n (-1)^{(i-t)} \binom{n-t}{i-t} \alpha_t &= \sum_{t=0}^n (-1)^{(i-t)} \binom{n-t}{i-t}  \sum_{j=0}^n \binom{n-j}{t-j} \beta_j \\
    &=\sum_{t=0}^n (-1)^{(i-t)}  \sum_{j=0}^n \binom{n-t}{n-i} \binom{n-j}{n-t} \beta_j \\
    &=\sum_{t=0}^n (-1)^{(i-t)}  \sum_{j=0}^n \binom{i-j}{i-t} \binom{n-j}{n-i} \beta_j \\
    &=  \sum_{j=0}^n  \binom{n-j}{n-i} \beta_j \sum_{t=0}^i (-1)^{(i-t)}\binom{i-j}{i-t} \\
    &=  \sum_{j=0}^n  \binom{n-j}{n-i} \beta_j \sum_{s=0}^{i-j} (-1)^{s}\binom{i-j}{s} \\
    &=  \beta_i,
\end{align*}
because $$\displaystyle\sum_{s=0}^{i-j} (-1)^{s}\binom{i-j}{s} =0$$ unless $i=j$, by the Binomial Theorem.
\end{proof}

We can finally established 
the MacWilliams-type identity for the refined weight distribution of a code. 

\begin{theorem}\label{thm:w_i^S}
   Let $\mC \le \F_q^n$ and fix $S \subseteq [n]$. For all $0 \le i \le n$ we have
\begin{multline}\label{eqt:MacWilliams} W_i^S(\mC) = q^{k-n - |S|} (q-1)^{|S|} \sum_{t=|S|}^n \sum_{j=0}^n \sum_{D \subseteq S} \sum_{B \subseteq D}  (-1)^{i-t+|D|} q^t\\ (q-1)^{-|B|}\binom{n-t}{i-t} \binom{n-|S|-j+|B|}{t - |S|}W_j^D(\mC^{\perp}).\end{multline}
\end{theorem}

\begin{proof}
The statement follows directly from \cref{cor:sumw_i^S} and \cref{lem:bomb}, by setting $\beta_i = W_i^S(\mC)$ and $$\alpha_t = q^{k-n+t-|S|}(q-1)^{|S|}\sum_{j=0}^n  \sum_{D \subseteq S}\sum_{B \subseteq D} (-1)^{|D|-|B|} (1 - q)^{-|B|} \binom{n-|S|-j+|B|}{t-|S|}W_j^D(\mC^{\perp}). \qedhere$$
\end{proof}

\begin{remark}
Although we call the above a MacWilliams-type identity, enumerators we consider do not naturally fit in the framework of the MacWilliams identities, as they do not represent the cardinalities of the blocks of a partition of the underlying code; see~\cite{gluesing2015fourier}. However,
one can see the above as a generalization of the classical MacWilliams identity. In fact, with tedious but straightfowrad computations one can show that  \eqref{eqt:MacWilliams} reduces to the famous MacWilliams identity~\cite{macwilliams1963theorem} when $S = \emptyset$.
\end{remark}

When applying MacWilliams-type identities, it is crucial to understand the dependencies between variables 
and which groups of these determine each other,
possibly under extra assumptions.
The next result shows
that $W_i^S(\mC)$ can be rewritten in terms of $W_j^T(\mC^{\perp})$ for $1 \leq j \leq n$ and all $T \subseteq [n]$ with $|T| = |S|$, provided that the cardinality of $S$ does not exceed the minimum distance of $\mC^\perp$. This result will become especially useful for the linear programming bounds presented in \cref{sec:Applications}.

\begin{corollary}\label{cor:w_i^Sfuldeter}
    Let $\mC \leq \F_q^n$ be a code,  $d^{\perp}$ the minimum distance of $\mC^{\perp}$, and  $S \subseteq [n]$. If $|S| \leq d^{\perp}$, then 
    \begin{multline*}W_i^S(\mC) = \binom{n-|S|}{i-|S|}q^{k-n}(q-1)^{i} + q^{k-n-|S|}(q-1)^{|S|}\sum_{t=|S|}^n \sum_{j=d^{\perp}}^n \sum_{\substack{T \subseteq [n]\\|T|=|S|}} \sum_{D \subseteq S\cap T}\sum_{B \subseteq D}(-1)^{i-t+|D|}\\ q^t(q-1)^{-|B|} \binom{n-t}{i-t}\binom{n-|S|-j+|B|}{t-|S|}\binom{j-|D|}{|S|-|D|}^{-1}W_j^T(\mC^{\perp}).
    \end{multline*}
\end{corollary}

\begin{proof}
Let $s:= |S|$ and consider the identity from Theorem~\ref{thm:w_i^S}. We first isolate the case $j=0$. Note that $W_0^D(\mC^{\perp}) =1$ if $D = \emptyset$ and $W_0^D(\mC^{\perp}) = 0$ otherwise. Hence the summand of Theorem~\ref{thm:w_i^S} corresponding to $j=0$ is:
\begin{align*}
 \sum_{t = s}^n (-1)^{i-t}q^t\binom{n-t}{i-t}\binom{n-s}{n-t}
=& \sum_{t = s}^n (-1)^{i-t}q^t\binom{n-s}{n-i}\binom{i-s}{i-t}\\
=& \binom{n-s}{n-i}\sum_{t=s}^{i}(-1)^{i-t}q^t\binom{i-s}{i-t}\\
=& \binom{n-s}{n-i}q^s \sum_{t=0}^{i-s}(-1)^{i-s-t}q^t\binom{i-s}{t}\\
=& \binom{n-s}{n-i}q^s (q-1)^{i-s},
\end{align*}
where 
the former equality follows from
the identity $$\binom{a}{b+c}\binom{b+c}{b} = \binom{a}{c}\binom{a-c}{b},$$ and the latter one follows from the Binomial Theorem. 

Now assume $j \geq 1$. If $j \leq d^{\perp}$ then $W_j^D(\mC^{\perp}) = 0$. Thus it suffices to consider the summands for $d^{\perp} \leq j \leq n$. By assumption $s \leq d^{\perp}$. Hence for all $D \subseteq S$ we can apply Lemma~\ref{lem:countS} and get 
$$W_j^D(\mC^{\perp}) = \binom{j-|D|}{s-|D|}^{-1}\sum_{\substack{D \subseteq T \subseteq [n]\\|T|=s}}W_j^T(\mC^{\perp}).$$
Therefore for $j \geq d^{\perp}$, \eqref{eqt:MacWilliams} can be rewritten as
\begin{small}
\begin{align*}
&\sum_{j=d^{\perp}}^n \sum_{t=s}^n \sum_{D \subseteq S} \sum_{B \subseteq D} (-1)^{i-t+|D|}q^t(q-1)^{-|B|}\binom{n-t}{i-t}\binom{n-s-j+|B|}{t-s} \binom{j-|D|}{s-|D|}^{-1}\sum_{\substack{D \subseteq T \subseteq [n]\\|T|=s}}W_j^T(\mC^{\perp})\\
=&\sum_{j=d^{\perp}}^n \sum_{t=s}^n \sum_{\substack{T \subseteq [n]\\|T|=s}} \sum_{D \subseteq S\cap T}\sum_{B \subseteq D}(-1)^{i-t+|D|}q^t(q-1)^{-|B|}\binom{n-t}{i-t}\binom{n-s-j+|B|}{t-s}\binom{j-|D|}{s-|D|}^{-1}W_j^T(\mC^{\perp}).
\end{align*}
\end{small}
Putting everything together proves the desired statement. 
\end{proof}

Note that in the above proof
the assumption $|S| \leq d^{\perp}$ was needed to apply \cref{lem:countS}. A natural question is whether $W_i^S(\mC)$ or not can be fully expressed in terms of $W_j^T(\mC^{\perp})$, for $0 \leq j \leq n$ and $|T| = |S|$ when $|S| > d^{\perp}$.
At the time of writing this paper we are unable to answer this question.

\section{Applications and Bounds}\label{sec:Applications}

Delsarte’s linear programming (LP) bound~\cite{delsarte1973algebraic} is 
a powerful
tool to estimate the size of a code with a certain length and minimum distance.
It combines the classical Macwilliams identities with, as the name suggests, linear programming. In this section we study a new~LP bound for codes with locality.
We use our main duality result,
Theorem~\ref{thm:w_i^S},
to build a linear program that gives a bound on the size of $(r,\delta)$-LRCs.

\begin{notation}
In the sequel, for a code $\mC \le \F_q^n$, $0 \le i \le n$ and $1 \le j \le n$, we will write $\smash{W_i^j(\mC)}$ instead of $\smash{W_i^{\{j\}}(\mC)}$.
\end{notation}

The next result shows that
it is possible
to bound from below, in any
$(r, \delta)$-LRC,
the number of codewords with coordinate $i$ in their support and weight at most $r+\delta-1$.

\begin{proposition} \label{prop:rdellrc}
    Let $\delta \geq 2$ and let $\mC \leq \F_q^n$ be an $(r, \delta)$-LRC, with dimension $k$. Then for all $j \in [n]$ we have
    $$\sum_{i=0}^{r+\delta -1} W^j_i(\mC^{\perp}) \geq q^{\delta-1} - q^{\delta -2}.$$
\end{proposition}

\begin{proof}
    Let $j \in [n]$ and $S_j \subseteq [n]$ be minimal with respect to inclusion
    such that $|S_j| \leq r +\delta -1$, $j \in S_j$, and $d(\pi_{S_j}(\mC)) \geq \delta$. Let $k_j := \dim(\pi_{S_j}(\mC))$, $d_j := d(\pi_{S_j}(\mC))$, and $k_j^{\perp} = |S_j| - k_j$. From the Singleton Bound we have $|S_j|-k_j+1\ge \delta$, which means $k_j^\perp \ge \delta-1$. Furthermore there exist $v \in \mC^{\perp}(S_j)$ such that $j \in \sigma(v)$, since otherwise we arrive to the contradiction that $\mC^{\perp}(S_j)$ would be degenerate and $d(\pi_{S_j}(\mC)) =1 < \delta$. 
    These two facts together imply that there are at least $q^{\delta-1} - q^{\delta - 2}$ codewords in  $\mC^{\perp}(S_j)$ containing $j$ in their support and of weight at most $|S_j| \leq r+ \delta -1$. 
    Therefore $\sum_{i=0}^{r+\delta-1} W^j_i(\mC^{\perp}) \geq q^{\delta-1} - q^{\delta - 2}$.
\end{proof}

Unfortunately, the converse of this statement is not true, unless $\delta=2$ (leading to the statement of Lemma~\ref{lem:charloc} below). Consider the following example. 

\begin{example}
Let $G$ and $H$ be matrices over $\F_2$ defined as follows:
$$G := \begin{pmatrix}1 & 1 & 0 & 1& 0 \\ 0 & 1 & 1 & 0 &0 \\ 0 &0 & 0&1 & 1 \end{pmatrix},
\quad H := \begin{pmatrix} 1 & 1 & 1 & 0 & 0 \\ 1 & 0 &0 & 1 & 1 \end{pmatrix}.$$
Let $\mC=\{xG : x \in \F_2^3\}$ which then implies $\mC^\perp=\{xH : x \in \F_2^2\}$. If we fix $r:=2$ and $\delta:=3$, then $r + \delta -1 = 4$ and $q^{\delta-1} - q^{\delta-2}  =2$. One can easily check that for all $j \in [5]$ we have 
$$\sum_{i=0}^{r+ \delta -1}W_i^j(\mC^{\perp}) =\sum_{i=0}^{4}W_i^j(\mC^{\perp}) = 2.$$
However, for the coordinate $\{1\}$ there exists no set $S_1 \subseteq [5]$ such that $1 \in S_1$, $|S_1| \leq 4$, and $d(\pi_{S_1}(\mC)) \geq 3$. In fact any projection of $\mC$ onto a set of coordinate of size less or equal to $4$ will have minimum distance $1$ or $2$. Hence the code $\mC$ is not an $(2, 3)$-LRC. 
\end{example}

As already mentioned, for $\delta=2$ the converse of Proposition~\ref{prop:rdellrc} is also true. As a byproduct, the result also gives a characterization of codes with locality.

\begin{lemma} \label{lem:charloc}
A non-degenerate code $\mC \le \F_q^n$ has locality $r$ if and only if $\sum_{i=2}^{r+1}W_{i}^j(\mC^\perp) \ge  q-1$ for all $1 \le j \le n$.
\end{lemma}

Proposition~\ref{prop:rdellrc} and Lemma~\ref{lem:charloc} indicate
that the MacWilliams-type identity of Corollary~\ref{cor:w_i^Sfuldeter} becomes especially useful in the context of LRCs when considering the weight distribution $W_i^{j}(\mC^{\perp})$, where $1 \leq j \leq n$ and $0 \leq i \leq n$. By applying \cref{cor:w_i^Sfuldeter}, we get the following corollary. 

\begin{corollary}\label{cor:w_i^ssingleelem}
    Let $\mC \leq \F_q^n$ be a code and  $l \in [n]$. Then 
    \begin{multline*}W_i^l(\mC) = \binom{n-1}{i-1}q^{k-n}(q-1)^{i} + q^{k-n-1}(q-1) \sum_{j=d^{\perp}}^n \sum_{s=1}^n \sum_{t=1}^i (-1)^{i-t}q^t\binom{n-t}{i-t}\\ \left(\frac{1}{j}\binom{n-1-j}{t-1}\right)^{1-\delta(l,s)} \left( \frac{1-j}{j}\binom{n-1-j}{t-1} - (q-1)^{-1}\binom{n-j}{t-1} \right)^{\delta(l,s)}W_j^s(\mC^{\perp}),
    \end{multline*}
    where $\delta(l,s) =1$ if $l=s$ and $0$ otherwise. 
\end{corollary}

\begin{proof}
Using Corollary~\ref{cor:w_i^Sfuldeter} for $S = \{l\}$  we immediately get 
  \begin{multline*}W_i^{l}(\mC) = \binom{n-1}{i-1}q^{k-n}(q-1)^i + q^{k-n-1}(q-1)\sum_{t=1}^n \sum_{j=d^{\perp}}^n \sum_{\substack{T \subseteq [n]\\|T|=1}} \sum_{D \subseteq \{l\}\cap T}\sum_{B \subseteq D}(-1)^{i-t+|D|}q^t\\(q-1)^{-|B|}\binom{n-t}{i-t}\binom{n-1-j+|B|}{t-1}\binom{j-|D|}{1-|D|}^{-1}W_j^T(\mC^{\perp}).
    \end{multline*}
We now compute 
\begin{equation}\label{equ:w_i^l}\sum_{D \subseteq \{l\}\cap T}\sum_{B \subseteq D}(-1)^{|D|}(q-1)^{-|B|}\binom{n-1-j+|B|}{t-1}\binom{j-|D|}{1-|D|}^{-1}, \end{equation}
for both $T = \{l\}$ and $T = \{s\}$ where $s \neq l$. In the former case, \eqref{equ:w_i^l} is equal to
\begin{align*}
    &\frac{1}{j}\binom{n-1-j}{t-1} - \left[(q-1)^{-1}\binom{n-j}{t-1} + \binom{n-1-j}{t-1}\right].
\end{align*}
In the case $T= \{s\}$ where $s \neq l$, \eqref{equ:w_i^l} is equal to $\frac{1}{j}\binom{n-1-j}{t-1}$. The results then follows. 
\end{proof}

\begin{remark}
    There exist other methods to derive the MacWilliams-type identity of \cref{cor:w_i^ssingleelem}. A first approach was established in \cite{gruica2022duality}, where
    the parameter $W_i^l(\mC)$ is rewritten in terms of the $W_j(\mC)$'s and  $W_j(\mC([n]\setminus\{l\}))$'s for $1 \leq j, \, l \leq n$.  The classical MacWilliams identity is then applied to the latter terms. In addition, it was brought to our attention by M. Grassl that our refined weight distribution seems related to the notion of \textit{split weight enumerator}; see \cite{simonis1995macwilliams} for more details. We were able to show that this is in fact the case and are able 
    to derive 
    \cref{cor:w_i^ssingleelem} also as a corollary
    of~\cite[Proposition 1]{simonis1995macwilliams}. 
\end{remark}

The last result we need to apply a linear program is the following, which is a specific case of \cref{lem:countS} for $A = \emptyset$ and $|S| = 1$.

\begin{lemma} \label{lem:helpdual}
Let $\mC \le \F_q^n$ be non-degenerate. For all $1 \le i \le n$ and $1\le j \le n$ we have
$$W_i(\mC) = {\displaystyle\sum_{j=1}^n W_i^j (\mC)}/{i}.$$
\end{lemma}

For ease of exposition, we introduce the following notation.

\begin{notation} \label{not:perp}
Let $A=\{a_{il} \mid 1\le i \le n+1, \, 1\le l \le n\} \subseteq \R_{\ge 0}$. For $1 \le i,l \le n$ we denote by $a_{il}^\perp$ the following linear combination of elements of $A$:
\begin{multline*}
    a_{il}^\perp =
    \binom{n-1}{i-1}(q-1)^{i} + q^{-1}(q-1) \sum_{j=d^{\perp}}^n \sum_{s=1}^n \sum_{t=1}^i (-1)^{i-t}q^t\binom{n-t}{i-t} \cdot \\ \left(\frac{1}{j}\binom{n-1-j}{t-1}\right)^{1-\delta(l,s)} \left( \frac{1-j}{j}\binom{n-1-j}{t-1} - (q-1)^{-1}\binom{n-j}{t-1} \right)^{\delta(l,s)}a_{js}.
\end{multline*}
\end{notation}

The following result
gives a linear program that establishes
bounds for codes with $(r,\delta)$-locality. 

\begin{theorem}[LP bound for $(r,\delta)$-LRCs] \label{thm:lpbound}
Let $\mC \le \F_q^n$ be a non-degenerate code of minimum distance at least $d$, dimension $k$, and $(r,\delta)$-locality. 
Let $\mu^*$ denote the minimum
value of 
$$\sum_{i=1}^{n}\left({\displaystyle\sum_{j=1}^{n} a_{ij}}/{i}\right),$$
where $a_{ij} \in \R$, for $1 \le i \le n+1$ and $1 \le j \le n$,
satisfy the following constraints:
\begin{itemize}
    \vspace{0.15cm}\item[(i)]  $a_{ij}\ge 0$ for $1 \le i,j \le n$, 
    \vspace{0.15cm}\item[(ii)]  $a_{ij}^\perp\ge 0$ for $1 \le i,j \le n$, 
    \vspace{0.15cm}\item[(iii)] $a_{ij}^\perp =0$ for $1 \le i \le d-1$ and $1 \le j \le n$,   
    \vspace{0.15cm}\item[(iv)]  $\smash{\textstyle\sum_{i=1}^{r+\delta-1}a_{ij} \ge q^{\delta-1}-q^{\delta-2}}$ for $1 \le j \le n$,
    \vspace{0.15cm}\item[(v)] $a_{1j} =0$ for $1 \le j \le n$.
\end{itemize}
Then
$$k \le n- \lceil \log_q(1+\mu^*) \rceil.$$
\end{theorem}
\begin{proof}
We claim that for any non-degenerate linear code $\mC \le \F_q^n$ of minimum distance $d$ and $(r,\delta)$-locality, the assignment $\smash{a_{ij} = W_i^j(\mC^\perp)}$ is a feasible solution of the linear program. Indeed, (i) and~(ii) are satisfied trivially. Constraint~(iii) guarantees that the minimum distance of the code is at least~$d$ and constraint~(iv) needs to be fulfilled in order for the code to be an~$(r,\delta)$-LRC. Finally constraint~(v) makes sure that $\mC$ is non-degenerate. Therefore the optimum of the linear program gives a lower bound on the size of $\mC^\perp$ by Proposition~\ref{prop:rdellrc}.
\end{proof}

We conclude this section with some results obtained 
with the LP bound of Theorem~\ref{thm:lpbound}
in Tables~\ref{table_LPbound1}-\ref{table_LPbound5}. 
For fixed parameters $(q,n,d,r,\delta)$, we give
an upper bound on the dimension~$k$ of an $(r,\delta)$-LRC in $\F_q^n$ and minimum distance at least $d$. The computations were performed using \texttt{SageMath}. In the case where $\delta =2$ (i.e. Tables \ref{table_LPbound1}-\ref{table_LPbound4}) we compare our results to \eqref{eqt:single} (\textbf{gen. Singl}) and \eqref{eqt:kopt} (\textbf{SH}).
The latter is computed in two different ways: first, the value~$\smash{k_{\textrm{opt}}^{(q)}}$ is estimated using a classical LP bound (\textbf{SH with LP}); second, the exact value of~$\smash{k_{\textrm{opt}}^{(q)}}$ is computed (\textbf{SH exact}), since it can be determined for the parameters considered in the tables by looking at databases of codes. 
The most fair comparison between our result and  \eqref{eqt:kopt} is given by SH with LP, since the value of SH exact is only known for very small parameters. 
Our results show that \cref{thm:lpbound} is often
tighter than \eqref{eqt:kopt} (this is indicated in red on the tables). A similar comparison is done in \cref{table_LPbound5}, but using  \eqref{eqt:Singlrdelta} and \eqref{eq:rdeltakopt} instead. 

\begin{small}
\begin{table}[h!]
\caption{$q=2$, $\delta =2$}
    \centering
    \renewcommand\arraystretch{1.1}
\begin{tabular}{|c|c|c|c|c|c|c|} 
 \hline
$n$ & $d$ & $r$ & LP & SH with LP  & SH exact & gen. Singl.
  \\\noalign{\global\arrayrulewidth 1.8pt}
    \hline
    \noalign{\global\arrayrulewidth0.4pt}
10 & 4 & 2 & $k \leq \red{4}$ & $ k \le 5$ & $k \le 5$ & $k \le 5$\\
\hline
11 & 3 & 3 & $k \le \red{6}$ & $k \le 7$ & $k \le 7$ & $k \le 7$\\
\hline
12 & 3 & 3 & $k \le \red{6}$ & $k \le 7$ & $ k \le 7$ & $k \le 7$\\
\hline
14 & 4 & 6 & $k \le \red{8} $ & $k \le 9$ & $k \le 9$ & $k \le 10$\\
\hline 
17 & 7 & 2 & $k \le \red{5}$ & $k \le 6$ & $ k \le 6$ & $k \le 8$ \\
\hline
18 & 8 & 2& $k\le \red{5}$ & $k \le 6$ & $k \le 6$ & $k \le 8$\\
\hline
20 & 9 & 5 & $k \le 5$ & $k \le 6$ & $k \le 5$ & $k \le 11$\\
\hline 
20 & 11 & 2 & $k \le \red{2}$ & $k \le 3$ & $k \le 3$ & $k \le 7$\\
\hline 
\end{tabular}
\label{table_LPbound1}
\end{table}

\begin{table}[h!]
\caption{$q=3$, $\delta =2$}
    \centering
    \renewcommand\arraystretch{1.1}
\begin{tabular}{|c|c|c|c|c|c|c|} 
 \hline
$n$ & $d$ & $r$ & LP & SH with LP  & SH exact & gen. Singl.
  \\\noalign{\global\arrayrulewidth 1.8pt}
    \hline
    \noalign{\global\arrayrulewidth0.4pt}
10 & 2 & 4 & $k \le \red{7}$ & $k \le 8$ & $k \le 8$ & $k \le 8$\\
\hline
11 & 6 & 4 & $k \le \red{4}$ & $k \le 5$ & $k \le 5$ & $k \le 5$\\
\hline
11 & 5 & 5 & $k \le \red{5}$ & $k \le 6 $ & $k \le 6$ & $k \le 6$\\
\hline 
13 & 9 & 2 & $k \le \red{2}$ & $k \le 3$ & $ k \le 3$ & $k \le 4$\\
\hline
14 & 2 & 6 & $k \le \red{11}$ & $k \le 12$ & $k \le 12$ & $k \le 12$\\
\hline
14 & 9 & 3 & $k \le 3$ & $k \le 4$ & $k \le 3$ & $k \le 5$\\
\hline
18 & 6 & 5 & $k \le 10$ & $k \leq 11$ & $k \le 10 $ & $k \le 11$\\
\hline
25 & 5 & 5 & $k \le \red{15}$ & $k \le 17$ & $k \le 17$ & $k \le 18$\\
\hline
\end{tabular}
\label{table_LPbound2}
\end{table}

\begin{table}[h!]
\caption{$q=4$, $\delta =2$}
    \centering
    \renewcommand\arraystretch{1.1}
\begin{tabular}{|c|c|c|c|c|c|c|} 
 \hline
$n$ & $d$ & $r$ & LP & SH with LP  & SH exact & gen. Singl.
  \\\noalign{\global\arrayrulewidth 1.8pt}
    \hline
    \noalign{\global\arrayrulewidth0.4pt}
    9 & 3 & 3 & $k \le \red{5}$  &$k \le 6$ & $k \le 6$ & $k \le 6$\\
    \hline
    10 & 2 & 4 & $k \le \red{7}$ & $k \le 8$ & $k \le 8$ & $k \le 8$\\
    \hline
    11 & 8 & 2 & $k \le 2$ & $k \le 3$ & $k \le 2$ & $k \le 3$\\
    \hline
    12 & 2 & 5 & $k \le \red{9}$ & $k \le 10$ & $ k \le 10$ & $ k\le 10$\\
    \hline
    14 & 2 & 6 & $k \le \red{11}$ & $k \le 12$ & $k \le 12$& $k \le 12$\\
    \hline 
    15 & 10 & 4 & $k \le 4$ & $k \le 5$ & $k \le 4$ & $k \le 5$\\
    \hline
   15 & 8 & 6& $k \le 6$ & $k \le 7$ & $k \le 6$ & $ k \le 7$\\
   \hline
    20 & 2 & 9 & $k \le \red{17}$ & $k \le 18$ & $k \le 18$ & $k \le 18$\\
    \hline
\end{tabular}
\label{table_LPbound3}
\end{table}

\begin{table}[h!]
\caption{$q=5$, $\delta =2$}
    \centering
    \renewcommand\arraystretch{1.1}
 \begin{tabular}{|c|c|c|c|c|c|c|} 
 \hline
$n$ & $d$ & $r$ & LP & SH with LP & SH exact & gen. Singl.
  \\\noalign{\global\arrayrulewidth 1.8pt}
    \hline
    \noalign{\global\arrayrulewidth0.4pt}
    9 & 3 & 3 & $k \le \red{5}$ & $k \le 6$ & $k \le 6$  & $k \leq 6$\\
    \hline
    11 & 3 & 4 & $k \le \red{7}$ & $k \le 8$ & $k \le 8$  & $k \leq 8$\\
    \hline
    14 & 2 & 6 & $k \le \red{11}$ & $k \le 12$ & $ k\le 12 $ & $k \le 12$\\
    \hline
    16 & 2 & 7 & $k \le \red{13}$ & $k \le 14$ & $k \le 14 $ & $k \le 14$\\
    \hline 
    16 & 2 & 7 & $k \le \red{13} $& $k \le 14$ & $k \le 14$ & $k \le 14$\\
    \hline
    18 & 2& 8 & $k \le \red{15}$ & $k \le 16$ & $k \le 16 $ & $k \le 16$\\
    \hline
    22 & 2 & 10 & $k \le \red{19}$ & $k \le 20$ & $k \le 20 $ & $k \le 20$\\
    \hline
    24 & 2 & 11 & $k \le \red{21}$ & $k \le 22$ & $k \le 22$ & $k \le 22$\\
    \hline
\end{tabular}
\label{table_LPbound4}
\end{table}

\begin{table}[h!]
\centering
\caption{$q=2$, $\delta =3$}
\begin{tabular}{|c|c|c|c|c|c|c|} 

 \hline
$n$ & $d$ & $r$ & LP & SH with LP & SH exact & gen. Singl
  \\\noalign{\global\arrayrulewidth 1.8pt}
    \hline
      \noalign{\global\arrayrulewidth0.4pt}
    15 & 7 & 5 & $k \leq 5$ & $k \leq 5$ & $k \leq 5$ & $k \leq 7$  \\ \hline
    16 & 5 & 5 & $k \leq 8$ & $k \leq 7$ & $k \leq 7$ & $k \leq 10$  \\ \hline
    16 & 8 & 5 & $k \leq \red{4}$ &  $k \leq 5$ & $k \leq 5$ & $k \leq 7$ \\ \hline
    17 & 9 & 5 & $k \leq 3$ & $k \leq 3$ & $k \leq 3$ &$k \leq 7$  \\ \hline
    17 & 10 & 5 & $k \leq 2$ & $k \leq 2$ & $k \leq 2$ & $k \leq 6$ \\ \hline
    18 & 7 & 5 & $k \leq 7$ & $k \leq 7$ & $k \leq 7$ & $k \leq 10$  \\ \hline
    19 & 8 & 5 & $k \leq 7$ & $k \leq 7$ & $k \leq 7$ & $k \leq 10$  \\ \hline
    20 & 5 & 7 & $k \leq 13$ & $k \leq 11$ & $k \leq 11$ & $k \leq 14$  \\ \hline
    \noalign{\global\arrayrulewidth0.4pt} 
\end{tabular}
  \label{table_LPbound5}
\end{table}
\end{small}

\begin{remark}
As Table~\ref{table_LPbound5} illustrates, 
when $\delta > 2$ our LP bound does not seem to beat the shortening bound (both the exact and the LP one) as frequently as for $\delta =2$. A possible explanation for why this happens is that, for $\delta > 2$, the $(r,\delta)$-LRCs are not fully characterized by the constraint (iv) in \cref{thm:lpbound}. Hence a code that achieves the minimum of the objective function 
may not be the dual code of an $(r, \delta)$-LRC.
\end{remark}

\newpage
\section{The Role of the Field Size}\label{sec:fieldsize}

In this section we establish various bounds for LRCs that involve the underlying field size $q$. The goal is to understand the role played by this parameter for the code's locality.

We first use the refined weight distribution introduced in Section~\ref{sec:finerweight} to derive bounds connecting 
the dual distance with the locality and the field size. We then apply the theory of generalized weights to derive bounds involving the length, dimension, minimum distance, locality, and field size.

To derive our first bound for LRCs, recall
from Definition~\ref{def:cst} that
$$\smash{\mC(S, [n]) = \{x \in \mC : S \subseteq \sigma(x)\}}.$$
The following result holds for any linear code and does not depend on the locality parameter.

\begin{proposition} \label{prop:weightbound}
    Let $\mC \le \F_q^n$ be an code of dimension $k$ and minimum distance $d$. Let $S \subseteq [n]$ with $|S|\leq k-1$. If $|\mC(S,[n])| > 0$, then there exist $x \in \mC(S,[n])$ such that 
    \begin{equation}\label{eq:weightbound}
        \omega^{\mathrm{H}}(x) \leq n - k +|S| - (d-q)/q.
    \end{equation}
\end{proposition}

\begin{proof}
    Let $s=|S|$ and $x \in \mC(S,[n])$ of minimum weight with $x \ne 0$.
    Let $t= \wH(x)$. Without loss of generality, we may assume $S= \{1, \ldots , s\}$, $T= \{1, \ldots, t\}$, and $x = \sum_{i=1}^t e_i$, where $e_i$ is the $i$-th standard basis elements. Since $s \leq k-1$ we have $t \leq n - k + s$. Furthermore, since $\dim(\mC) = k$ then there must be an element $y \in \mC$ such that $y_i = 0$ for all $i \in S \cup \{n-k+s+2, \ldots, n\}$. Then $\wH(\pi_T(y)) \geq d- (n-t) + k-s-1$. Moreover, there exist $\alpha \in \fq^*$ such that
    $$\{j \, : \, \pi_T(y) = \alpha\} \geq (d-n+t+k -s-1)/(q-1).$$
    Let $x' = \alpha x - y$. Then $\alpha_i \neq 0$ for all $i \in S$ and 
   \begin{align*}
       \wH(x') &\leq t - \frac{d-n+t+k-s-1}{q-1} + (n-t) - (k-s-1)\\
                &= n-k+s+1 - \frac{d-n+t+k-s-1}{q-1}.
   \end{align*}
   By choice of $t$ we therefore have
   $$t \leq n-k+s+1 - \frac{d-n+t+k-s-1}{q-1},$$
   which is equivalent to $t \leq n - k +s - (d-q)/q$. \
\end{proof}

By applying the previous proposition to a set $S$ of cardinality 1, we obtain the following corollary, which was presented also in \cite{gruica2022duality}.

\begin{corollary}\label{cor:sharperSinglbound}
    Let $\mC \leq \F_q^n$ of dimension $k$ and minimum distance $d$. Then 
    $$d \leq n - k +1  - (d-q)/q.$$
\end{corollary}

Note the above result improves on the Singleton bound for $d \geq q$. Furthermore, it shows that for linear MDS codes one must have $d \leq q$.
In the context of 
\cref{cor:sharperSinglbound}, one may wonder if the bound of \cref{prop:weightbound}, when applied to a set $S$ with $2 \leq |S| \leq k-1$, can lead to an upper bound
for the $|S|$-generalized weight (see Definition~\ref{def:genweights}). In other words: Is it true that  $d_{|S|} \le n- k+|S| - (d-q)/q$?   
Although there exist examples of codes for which the inequality holds, this is in general not the case, as illustrated by the following example.

\begin{example}\label{ex:simplex}
Let $\mC \leq \F_2^7$ be the simplex code of dimension $3$, i.e., the code generated by 
$$ G = \begin{pmatrix} 
    1 & 0 & 0 & 1 & 0 & 1 & 1 \\
    0 & 1 & 0 & 1 & 1 & 0 & 1 \\
    0 & 0 & 1 & 0 & 1 & 1 & 1 
\end{pmatrix}.$$
It is well known that $\wH(x) = 4$ for all non-zero $x \in \mC$. Therefore for all $S \subseteq [7]$ such that $|S| = 2$, the smallest weight of a codeword  $x \in \mC(S,[n])$ is $4$. 
One can easily show that $d_2(\mC) = 6$, which is greater than the right hand side of \eqref{eq:weightbound} when the parameters are those of the simplex code. 
\end{example}

By applying \cref{prop:weightbound} to the dual code $\mC^{\perp}$ of an LRC $\mC$, one can easily derive the following bound. 

\begin{corollary}
Let $\mC \leq \F_q^n$ 
be a code
with dimension $k \ge 2$, minimum distance $d$, and locality~$r$. Let $d^\perp=d(\mC^\perp)$. We have
\begin{align*}
    d^\perp -1+\frac{d^\perp-q}{q} \le n-(d-2)-\left\lceil \frac{k}{r}\right\rceil.
\end{align*}
\end{corollary}
\begin{proof}
Note that if a code $\mC$ has locality $r$, then by Lemma~\ref{lem:helpdual} there exists a codeword $x \in \mC^\perp$ with $\wH(x) \le r+1$. In particular, $d^\perp \le r+1$. Applying Corollary~\ref{cor:sharperSinglbound} we obtain
\begin{align*}
    d^\perp \le k+1-(d^\perp-q)/q
\end{align*}
which, combined with Theorem~\ref{thm:single}, yields
\begin{align*}
    d^\perp \le n-(d-2)-\left\lceil \frac{k}{r}\right\rceil+1-(d^\perp-q)/q.
 \end{align*}
This concludes the proof.
\end{proof}

We now use the \textit{generalized weights} of a code to derive new bounds which we later apply to LRCs. We start by recalling the needed definitions.

\begin{definition}\label{def:genweights}
Let $\mC \leq \F_q^n$ be a code. The \textbf{$i$-th generalized weight} of $\mC$ is $$d_i(\mC) := \min\{|\sigma(\mD)| \, : \mD \leq \mC \textup{ and } \dim(\mD) = i\}$$ for $1 \leq i \leq k$, and where $\sigma(\mD):=\bigcup_{x \in \mD} \sigma(x)$.  \end{definition}

It is easy to see that $d_1 = d$.
Moreover, the generalized weights of a code are strictly increasing, i.e., $d_1(\mC) < d_2(\mC) < \ldots < d_k(\mC)$; see e.g.~\cite{wei1991generalized}. The following are well-known bounds involving the generalized weights of codes established respectively in \cite{wei1991generalized} and \cite[Theorem 1]{helleseth1995bounds}. 

\begin{theorem}\label{thm:boundsgenweight}
    Let $\mC \leq \F_q^n$ be a code of dimension $k$ and let $d_1, \ldots, d_k$ be its generalized weights. Then:

    \begin{itemize} 
        \item[(i)] $d_i \leq n - k + i$,
        \item[(ii)] $(q^i -1)d_{i-1} \leq (q^i  -q)d_i$.
    \end{itemize}
\end{theorem}

Part (i) of Theorem~\ref{thm:boundsgenweight} is often referred to as the \textit{generalized Singleton-type bounds}.

The first two generalized weights of a code provide extremely useful information when trying to determine the smallest underlying field size over which some classes of codes
exist. 
This is shown in the next result. It is, to the best of our knowledge, the first time such a bound is derived.

\begin{theorem}\label{thm:boundonq}
    Let $\mC \leq \F_q^n$ be a code of minimum distance $d$. If $d_2 = d + s$, then 
    \begin{equation}\label{eqt:boundonq}
    d \leq sq.
    \end{equation}    
\end{theorem}

\begin{proof}
    By \cref{thm:boundsgenweight} (ii) we know that $(q^2 -1)d \leq (q^2-q)d_2 = (q^2-q)(d+s)$. Rewriting this equality gives us $d \leq sq$, as desired. 
\end{proof}

As a simple consequence of
the previous theorem, we get the following bound. 

\begin{corollary}
    Let $\mC \leq \F_q^n$ have dimension $k$ and minimum distance $d$. Then 
    \begin{equation}
        d \leq \frac{q}{q+1}(n-k+2).
    \end{equation}
\end{corollary}
\begin{proof}
    Let $d_2$ be the second generalized weight of $\mC$. Then by Theorem~\ref{thm:boundsgenweight} (i) we have $d_2 -d \leq n-k+2 - d$.
    Moreover, by applying \cref{thm:boundonq} we get
    $$d \leq q(n-k+2 -d).$$
    The inequality above can then 
    be rewritten a $d \leq \frac{q}{q+1}(n-k+2)$, which proves the wanted result. 
\end{proof}

In the remaining part of this section, we study the generalized weights of an LRC. 
While these parameters were already considered in \cite{hao20}, 
we propose a new approach 
based on a 
new class of code parameters.
The latter turn out to be closely related to the locality, and allow to provide concise proofs for known  results and extend them as well.

\begin{notation}\label{not:mui}
    Let $\mC \leq \F_q^n$ be of dimension $k$, $\mC^{\perp}$ its dual code of dimension $k^{\perp}=n-k$ with $i$-th generalized weight $d_i^{\perp}$. Define
    $$\mu_i(\mC) := \min\{t \, : \, d^{\perp}_t \geq n - k^{\perp} - (i-1) + t\}.$$
\end{notation}

The parameter $\mu_i(\mC)$ captures the smallest dimension $t$ such that all subcodes of $\mC^{\perp}$ of dimension $t$ have support of size at least  $n - k^{\perp} - (i-1) + t$. In other words, $\mu_i(\mC)$ is the smallest dimension of subscodes needed such that the generalized weight of the dual code of said dimension has defect at least $(i-1)$ from the generalized Singleton bound.

\begin{remark}
    The parameters $\mu_i$ can be seen as a generalization of the parameter $\mu$ introduced in \cite{tamo2016optimal}. In the latter the parameter $\mu$ is defined in terms of matroid theory terminology but gives \cref{not:mui} when translated into coding theory language. 
    Furthermore, the parameters $\mu_i$ can also be defined for matroids, however for readability purposes we restrict to defining those parameters in terms of coding theory terminology. 
\end{remark}

Before presenting the next result, we recall the following well-known facts; see e.g.~\cite{huffman_pless_2003}.

\begin{lemma}\label{lem:punct/short}
    Let $\mC \leq \F_q^n$ have dimension $k$ and let $S \subseteq [n]$. Then
    \begin{itemize}
        \item[(1)] $\dim(\mC(S)) + \dim(\pi_{S^c}(\mC)) = k$. 
        \item[(2)] $ \pi_S (\mC)^{\perp} = \pi_S (\mC^{\perp}(S))$
    \end{itemize}
\end{lemma}

As the next result shows, the parameters $\mu_i$ are closely related to the generalized weights of the code $\mC$. The following proof is inspired by the work of \cite{tamo2016optimal}, in which similar methods were used to prove the same result for $i=1$.

\begin{theorem}\label{thm:mugenweight}
Let $\mC \leq \F_q^n$ be a code of dimension $k$ and let $d_i(\mC)$, for $1 \leq i \leq k$, be its generalized weights. Then 
$$d_i(\mC) = n - k - \mu_i(\mC) + i +1.$$
\end{theorem}

\begin{proof}
For notation purposes, throughout the proof we let $\mu_i= \mu_i(\C)$, $d_i = d_i(\mC)$, and $d_i^{\perp} = d_i(\mC^{\perp})$. For ease of exposition, we divide the proof into two claims.

\begin{claim} \label{cl:1}
We have $\mu_i \geq n - k - d_i +i +1$.
\end{claim}
\begin{clproof}
We show that $d^{\perp}_{k^{\perp} - d_i + i} < n-d_i+1$ by constructing a subspace of $\mC^{\perp}$ that has the desired dimension and support. 
Let $\mD \leq \mC$ and $S:= \sigma(\mD)$ with the property that $\dim(\mD) = i$ and $|S| = d_i$.
Since $\mD$ achieves the generalized weight, it must be that $\mD = \mC(S)$ and by \cref{lem:punct/short}(1) we know that $\dim(\pi_{S^c}(\mC))  = k-i$.
Furthermore,  using \cref{lem:punct/short}(2) we obtain
$\dim (\mC^{\perp}(S^c)) = \dim (\pi_{S^c}(\mC)^{\perp})  = n - d_i - k + i$.
Define $\mD' = \mC^{\perp}(S^c)$ and note that $\sigma(\mD') \leq n-d_i < n-d_i+1$. Therefore we have $d^{\perp}_{k^{\perp} - d_i +1} < n- d_i+1$, which establishes the statement.  
\end{clproof}

\begin{claim} \label{cl:2}
We have $\mu_i \leq n - k - d_i + i + 1$.
\end{claim}
\begin{clproof}
First assume $\mu_i = 1$. From Theorem~\ref{thm:boundsgenweight}(1), we have $d_i \leq n - k + i$ and therefore $\mu_i = 1 \leq n-k-d_i + i +1$, as desired.  

Now assume $\mu_i \geq 2$. Let $\mD \leq \mC^{\perp}$ such that $\dim(\mD) = \mu_i -1$ and $\sigma(\mD) =: S$ satisfies $|S| = d^{\perp}_{\mu_i -1} \leq k - i + \mu_i -1$. Such a subspace exists by definition of the parameter $\mu_i$. Furthermore, since $\dim(\mD) = \mu_i -1$, $\sigma(\mD) = S$ and $d_{\mu_i-1}^\perp = |S|$, we must have $\mD = \mC^{\perp}(S)$. 
Therefore using 
\cref{lem:punct/short}(2) we obtain:
\begin{align*}
\dim(\pi_S(\mC)) &= \dim(\pi_S(\mC^{\perp}(S))^{\perp})\\
                &= |S| - \dim(\mC^{\perp}(S)) \\
                &\leq  k -i + \mu_i - 1 - (\mu_i - 1)\\
                &= k-i.
\end{align*}
The above also gives that $|S| = \dim(\pi_S(\mC)) + \dim(\mD) = \dim(\pi_S(\mC)) + \mu_i -1$, which will be used later on.

Using \cref{lem:punct/short}(1), we moreover get  $s:= \dim(\mC(S^c)) = k - \dim(\pi_S(\mC)) \geq i$.
 Let $G \in \fq^{s \times n}$ be a
 generator matrix of $\mC(S^c)$
 in reduced row-echelon form (RREF). 
 Let $\mathcal{W} = \langle g_1, \ldots g_i \rangle$, where $g_i$ is the $i$-th row of $G$. Clearly, $\dim(\mathcal{W}) = i$, $D := \sigma(\mathcal{W}) \subseteq S^c$, and since $G$ is in RREF, $\mC(D) = \mathcal{W}$.
This  shows that $d_i \leq |D|$.
Now consider $A \subseteq S^c$ such that $A \cup D = S^c$ and $A \cap D = \emptyset$. Note that $D^c = S \cup A$. By \cref{lem:punct/short}(1), we have $\dim(\pi_{A \cup S}(\mC)) = k - \dim(\mC(D)) = k-i$.  This leads to the following chain of inequalities:
\begin{align*}
k-i &=  \dim(\pi_{A \cup S}(\mC))\\
    &\leq \dim(\pi_S(\mC)) + \dim(\pi_A(\mC)) \\
    &\leq \dim(\pi_S(\mC)) + |A|.
\end{align*}
Finally, we use the facts that $d_i \leq |D|$ and $|S| = \dim(\pi_S(\mC)) + \mu_i -1$  to get the following:
\begin{align*}
n- d_i &\geq n - |D|\\
        &= |S \cup A|\\
        &= |S| + |A|\\
        &\geq \dim(\pi_S(\mC)) + \mu_i -1 + k-i - \dim(\pi_S(\mC))\\
        &= \mu_i + k -i -1.
\end{align*}
This shows that $\mu_i \leq n- d_i - k +i+1$, as desired. 
\end{clproof}
\flushleft Combining Claim~\ref{cl:1} with Claim~\ref{cl:2} concludes the proof.
\end{proof}

Our next move is to show how the dimension and locality of a code directly impact the parameter $\mu_i$. The special case of $\mu_1$ was already shown in \cite{tamo2016optimal}.  

\begin{lemma}\label{bound mus}
Let $\mC \leq \F_q^n$ be an LRC of dimension $k$, locality $r$, generalized weights  $d_i$, and $\mu_i$ as in Notation~\ref{not:mui}. 
Then for all $1 \leq i \leq k$ we have
$$\mu_i \geq \left \lceil  \frac{k-(i-1)}{r} \right\rceil.$$
\end{lemma}

\begin{proof}
 Since $\mC$ has locality $r$, there exist at least $\left \lceil k/r \right \rceil -1$ linearly independent codewords $x_i \in \mC^{\perp}$, for $1 \leq i \leq \left \lceil k/r \right \rceil -1$, of weight at most $r+1$. Let $s:= \left \lceil (k - (i-1))/r \right \rceil -1 \leq \left \lceil k/r \right \rceil -1$. Consider $\mD = \langle  x_1, \ldots x_s \rangle$. Then 
 \begin{align*}
 \sigma(\mD) &\leq (r+1)\left(\left \lceil \frac{k-i+1}{r} \right \rceil -1 \right)\\
            &\leq k-i +1 +r -r +  \left \lceil \frac{k-i+1}{r} \right \rceil -1\\
            & < k- (i-1) + \left \lceil \frac{k-i+1}{r} \right \rceil.
\end{align*}
Hence $\mu_i >\left \lceil (k- (i-1))/r \right \rceil -1$, proving the lemma. 
\end{proof}

As a simple consequence we get the following characterization of optimal LRCs and a bound on the generalized weights of an LRC. 

\begin{corollary}
    Let $\mC \leq \F_q^n$ be a code with dimension $k$, minimum distance $d$ and locality~$r$. Furthermore, let $\mC^{\perp}$ be its dual code with generalized weight $\{d^{\perp}_1, \ldots, d^{\perp}_{n-k}\}$. The code $\mC$ is an optimal LRC if and only if $d^{\perp}_{\lceil k/r\rceil} = n- k^{\perp} + \lceil k/r \rceil$.
\end{corollary}

\begin{proof}
   By \cref{thm:mugenweight},  a code $\mC$ is an optimal LRC if and only if  $\mu_1 = \lceil k/r \rceil$. By definition, the latter is true if and only if $d_{\lceil k/r \rceil} \geq n - k^{\perp} + \lceil k/r \rceil$. However, by \cref{thm:boundsgenweight}(i) we also know that $d_{\lceil k/r \rceil} \leq n - k^{\perp} + \lceil k/r \rceil$. Hence equality must hold. 
\end{proof}

\begin{corollary}\label{thm:genweibound}
Let $\C \leq \F_q^n$ be an LRC of dimension $k$, locality $r$, and let $d_i$ denote its generalized weights where $1 \le i \le k$. Then 
$$d_i \leq n -k +i - \left(\left \lceil  \frac{k-(i-1)}{r} \right \rceil -1\right).$$
\end{corollary} 

\begin{proof}
Combining \cref{thm:mugenweight} and Lemma \ref{bound mus} we obtain
\begin{align*}
    d_i &= n-k - \mu_i +i+1 \\
        &\leq n -k - \left \lceil  \frac{k-(i-1)}{r} \right\rceil +i +1. \qedhere
\end{align*}
\end{proof}

\begin{remark}
 \cref{thm:genweibound} is a special case of the bound established in \cite[Theorem 1]{hao20} for $(r, \delta)$-LRCs. However, the result therein is proved using the gap numbers of a code, which differs from our approach. 
\end{remark}

Using the generalized Singleton-type bounds (see Theorem~\ref{thm:boundsgenweight} (i)), it is possible to determine or obtain a bound on the second generalized weight of optimal LRC.
Note that in \cite[Theorem 5]{hao20}, the authors establish the full generalized weight hierarchy of optimal LRCs under the assumption that $r|k$. 
The proof of \cite[Theorem~5]{hao20}
is based on the following result, which we can prove for completeness using the theory developed in this paper.

\begin{proposition}\label{prop:2genwei}
Let $\C \leq \F_q^n$ be an optimal LRC of dimension $k$, locality $r$, and let $d_i$ denote its $i$-th generalized weight, $1 \le i \le k$.
\begin{itemize}
    \item If $k \not \equiv 1 \pmod r$, then $d_2 = d_1 + 1$.
    \item If $k \equiv 1 \pmod r$, then $d_2 \leq d_1 +2$.
\end{itemize}
\end{proposition}

\begin{proof}
Assume $k \not \equiv 1 \pmod r$. First note that $\left \lceil k / r \right \rceil =\left \lceil (k-1) / r \right \rceil $. Since $\C$ is optimal, we have $d = n - k - \left \lceil k / r \right \rceil +2$. Using Theorem~\ref{thm:genweibound} we get the following bound on $d_2$:
\begin{align*}
    d_2 &\leq n- k +2 - (\left \lceil (k-1) / r \right \rceil) + 1\\
        &\leq n - k -  \left \lceil k / r \right \rceil + 3\\
        &= d_1 +1.
\end{align*}
However since $d_1 < d_2$  it must be that $d_2 = d_1 + 1$. 
Now assume $k \equiv 1 \pmod r$. Then $d_2 \leq d_1 +2$. Therefore $\left \lceil k / r \right \rceil =\left \lceil (k-1) / r \right \rceil -1 $. Using Theorem \ref{thm:genweibound} in a similar way as above, we get $d_2 \leq d_1 + 2.$
\end{proof}

Combining \cref{prop:2genwei} with \cref{thm:boundonq} we get as an immediate corollary an upper bound on the minimum distance of optimal LRCs. 
A similar result was shown in \cite[Theorem 2]{hao2020bounds} by  considering a suitable construction. We, on the other hand, show that the following bound can be seen as a direct consequence of the change in generalized weights of an optimal LRC.

\begin{corollary}\label{thm:boundq}
Let $\mC \leq \F_q^n$ be an optimal LRC of dimension $k$, minimum distance $d$, and locality $r$. 
\begin{itemize}
    \item If $k \not \equiv 1 \pmod r$, then $d \leq q$.
    \item If $k \equiv 1 \pmod r$, then $d \leq 2q$.
\end{itemize}
\end{corollary}

\begin{remark}
For $q=2$ there exist codes for which the bound in  \cref{thm:boundq} for $k \equiv 1 \pmod r$  is attained.
For example, the simplex code of \cref{ex:simplex} is an optimal LRC where  $d = 4 = 2q$.  
\end{remark}

We can now derive a Singleton-type bound for LRCs that depends on the field size. 

\begin{proposition}\label{prop:lrcsinglq}
    Let $\mC \leq \F_q^n$ be an LRC of dimension $k$ and locality $r$. Then 
    \begin{equation}
        d \leq \frac{q}{q+1}(n - k - \left \lceil \frac{k-1}{r} \right \rceil +3). 
    \end{equation}
\end{proposition}

\begin{proof}
    By \cref{thm:boundsgenweight} we have $d_2 \leq n - k - \lceil \frac{k-1}{r} \rceil +3$. Hence $\smash{d_2 - d \leq n - k - \left \lceil \frac{k-1}{r} \right \rceil +3 - d}$ and applying \cref{thm:boundonq} we get that 
    $$d \leq q(n - k - \left \lceil \frac{k-1}{r} \right \rceil +3 - d).$$
    By rewriting the inequality above we get the desired result. 
\end{proof}

The previous theorem leads to interesting restrictions on the parameters of optimal LRC, including the field size $q$. 

\begin{corollary}
    Let $\mC \leq \F_q^n$ be an optimal LRC with dimension $k$ and locality $r$. 
    \begin{itemize}
        \item If $k \leq q$ and $k \not \equiv 1 \pmod r$, then $n \leq 3q$, 
        \item If $k \leq q$ and $k \equiv 1 \pmod r$, then $n \leq 4q+1$. 
    \end{itemize}
\end{corollary}

\begin{proof}
   By rewriting \cref{prop:lrcsinglq} we have
    $$d \leq n- k - \left \lceil \frac{k-1}{r} \right \rceil +3 - \frac{n-k- \lceil (k-1)/r \rceil +3}{q+1}.$$
First we assume $k \not \equiv 1 \pmod r$. By applying the assumptions on $k$ we then obtain:
\begin{align*}
    d & \leq n- k - \left \lceil \frac{k}{r} \right \rceil +3 - \frac{n-q-\lceil (q-1)/r \rceil + 3}{q+1}\\
    & =  n- k - \left \lceil \frac{k}{r} \right \rceil +3 - \frac{n +3}{q+1} + \frac{q+\lceil (q-1)/r \rceil}{q+1}\\
    &=  n- k - \left \lceil \frac{k}{r} \right \rceil +2 - \frac{n+3}{q+1} +3
\end{align*}
Note that if  $3 -  ({n+3})/({q+1}) < 0$ then we get a sharper bound than \eqref{eqt:single}. Hence if $\mC$ is an optimal LRC we must have $3 - ({n+3})/({q+1}) \geq 0$, which is true if and only if $n \leq 3q$.

If $k \equiv 1 \pmod r$ we apply a similar argument with the only difference that $\lceil (k-1)/r \rceil = \lceil k/r \rceil -1$. We then get that $n \leq 4q+1$.
\end{proof}

\begin{remark}
    The previous result shows that
    an optimal LRC $\C \in \F_q^n$ cannot exist when $k \le q < (n-1)/4$.
\end{remark}

\newpage
\bibliographystyle{amsplain}
\bibliography{ourbib}

\end{document}